\documentclass[11pt]{article}
\usepackage[a4paper,margin=1in]{geometry}
\usepackage[T1]{fontenc}
\usepackage[utf8]{inputenc}
\usepackage{lmodern}
\usepackage{microtype}
\usepackage{amsmath,amssymb,amsthm,mathtools}
\usepackage{bm}
\usepackage{booktabs}
\usepackage{tabularx}
\usepackage{array}
\usepackage{enumitem}
\usepackage{graphicx}
\usepackage{xcolor}
\usepackage[numbers,sort&compress]{natbib}
\usepackage[colorlinks=true,linkcolor=blue!50!black,citecolor=blue!50!black,urlcolor=blue!50!black]{hyperref}

\newcommand{\R}{\mathbb{R}}
\newcommand{\E}{\mathbb{E}}
\newcommand{\Ivec}{\mathcal{I}}
\newcommand{\Exch}{\mathcal{E}}
\newcommand{\Pcal}{\mathcal{P}}
\newcommand{\one}{\mathbf{1}}
\newcommand{\e}{\mathbf{e}}
\newcommand{\KT}{d_{\mathrm{KT}}}

\newcommand{\barx}{\bar{x}_{ij}}
\newcommand{\deltax}{\delta_{ij}}

\DeclareMathOperator{\sign}{sign}
\DeclareMathOperator{\dist}{dist}

\newtheorem{definition}{Definition}[section]
\newtheorem{proposition}[definition]{Proposition}
\newtheorem{corollary}[definition]{Corollary}
\newtheorem{remark}[definition]{Remark}
\newtheorem{example}[definition]{Example}
\newtheorem{assumption}[definition]{Assumption}
\newtheorem{theorem}[definition]{Theorem}

\title{A Mathematical Theory of Ranking}
\author{\href{https://github.com/yinsn}{Yin Cheng}\thanks{Email: \href{mailto:yin.sjtu@gmail.com}{yin.sjtu@gmail.com}}}
\date{\today}

\begin{document}

\maketitle

\begin{abstract}
Ranking systems are typically implemented through scalar scores, yet the induced order depends only on pairwise comparisons rather than on the absolute scale of those scores. We develop a mathematical theory of ranking centered on \emph{pairwise margins}. In the linear case, each margin admits an exact factor-level decomposition; the resulting \(L_1\) local influence share is not merely a convenient normalization but the unique local budgeting rule compatible with pure factor refinement. Aggregating local shares gives a \emph{global influence structure}; in log-absolute-weight coordinates \(u_f=\log |w_f|\), this structure is the gradient of a convex potential, its Jacobian is a competition-graph Laplacian, and Influence Exchange satisfies a finite energy identity with a zero-exchange rigidity law. We extend the pairwise framework to nonlinear scoring through local linearization, path-based pairwise attribution, and Pairwise Integrated Gradients. There, factorwise path dependence is governed precisely by mixed-partial interaction curvature: factorwise pathwise uniqueness is recovered exactly in the additive/no-interaction regime. We further connect the continuous obstruction to discrete exactness, triangle curl, Hodge-like diagnostics on sampled pairwise graphs, and root-space and Weyl-chamber geometry. The formally proved core comprises canonical local budgeting, exchange geometry, and the nonlinear uniqueness boundary; the subsequent geometric perspectives serve as scoped interpretive closures rather than independently closed theorem systems.

\end{abstract}

\tableofcontents

\section{Introduction}

Ranking is typically implemented through scalar scores. In a recommender or retrieval system, the starting point is often a formula of the form
\begin{equation}
    s_i = \sum_{f=1}^d w_f x_{if},
\end{equation}
or, more generally, a nonlinear map \(s_i = F(x_i)\). While this score-first implementation is operationally convenient, the ranking itself depends only on relative order: the system ultimately determines whether item \(i\) should precede item \(j\), not whether \(s_i\) should attain one absolute value rather than another. The present paper develops this observation into a mathematical theory of ranking built on pairwise-first principles.

Our organizing concept is \emph{Influence Exchange}. Beginning with pairwise margins, the approach decomposes each local contest into factor-level contributions in the linear case, normalizes those contributions into local influence shares, aggregates them into a global influence structure, and compares that structure across model states. This yields a systematic accounting framework for ranking changes: rather than asking only how many pairwise relations changed, we ask how pairwise control was reallocated across factors.

The linear case serves as the foundational entry point because it admits exact factor-level decomposition without ambiguity. Nonlinear scoring is subsequently treated as a continuous extension of the same pairwise accounting problem, rather than as a separate explainability topic. Throughout, the theory maintains a clear separation between what is proved exactly, what follows from standard structural results, and what serves as geometric interpretation.

The mathematical development proceeds from exact linear decomposition, through canonical local budgeting, to rigid global exchange geometry, and finally to a sharp nonlinear boundary for factorwise uniqueness.

\begin{figure}[t]
    \centering
    \includegraphics[width=0.92\linewidth]{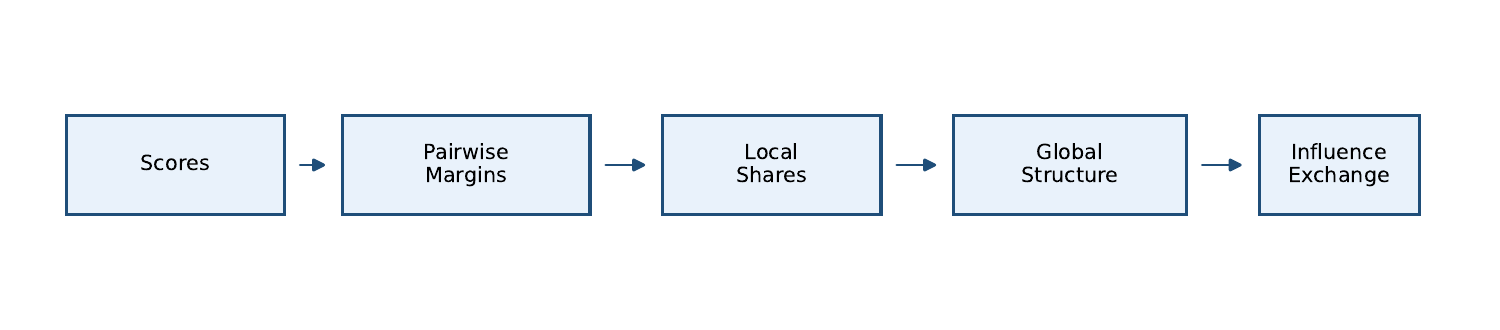}
    \caption{The conceptual pipeline of the paper. The analytical target moves from absolute scores to pairwise margins, local influence shares, global influence structure, and finally Influence Exchange across model states.}
    \label{fig:pipeline}
\end{figure}

\paragraph{Why pairwise margins.}
Pairwise ranking is well-established in learning to rank; pairwise losses and pairwise preference models are classical objects in the literature \citep{burges2005ranknet,burges2010ranknet,bradley1952rank,thurstone1927law}. Our claim is both narrower and distinct. We argue that \emph{pairwise margin} is the appropriate atomic object for influence accounting, and that centering the analysis on this object clarifies how local factor-level contributions, global allocation, and model-to-model change cohere within a single chain.

\paragraph{Linear exactness, nonlinear semantics.}
In the linear case, every pairwise margin decomposes exactly into factor-level local terms. This exactness makes it possible to define local influence shares and aggregate them into a global influence structure. Section~4 establishes that the local share is the unique refinement-consistent budgeting rule and that the resulting global structure possesses a convex-gradient, Laplacian-response geometry. In the nonlinear case, the pairwise margin remains well-defined; in general factorwise uniqueness fails, and Theorem~\ref{thm:interaction-curvature} identifies the exact no-interaction regime in which it is recovered. We therefore extend the framework through local linearization and path-based pairwise attribution, drawing on path methods such as Integrated Gradients and Aumann--Shapley-type value constructions \citep{sundararajan2017axiomatic,aumann1974values,lundberg2017unified}. The resulting framework maintains a clear distinction between what is formally determined and what constitutes a semantic choice.

\paragraph{Contributions.}
This paper makes four contributions.
\begin{enumerate}[leftmargin=2em]
    \item We recast ranking as a pairwise-order problem and identify pairwise margin as the atomic analytical object underlying flips, transpositions, and score-space boundary crossings.
    \item We prove a refinement-consistency characterization of the local influence share: once signed linear terms have been converted into nonnegative local efforts, the \(L_1\) share is the unique refinement-consistent local budget.
    \item We prove an exchange-geometry theorem: under log-absolute linear weights, the global influence structure is a convex-potential gradient, its response operator is a competition-graph Laplacian, and zero exchange is rigid up to componentwise common rescaling.
    \item We extend the pairwise accounting view to nonlinear scoring and prove an interaction-curvature theorem: factorwise path attribution is path-independent exactly when the relevant mixed partials vanish, with full factorwise uniqueness equivalent to additive separability. Section~8 then reinterprets the same obstruction in discrete curl and exactness language.
\end{enumerate}

\paragraph{Scope and boundaries.}
We do not claim to solve ranking explainability in general, nor do we claim a unique nonlinear attribution semantics. This manuscript is mathematical rather than a duplicate experimental report; however, the influence-accounting vocabulary is connected to production engineering practice. A companion Sortify technical report describes a fully autonomous ranking-optimization agent based on Influence Share/Exchange, reports deployment in a large-scale recommendation system, and documents positive online A/B outcomes followed by production rollout \citep{cheng2026letagentsteer}. The geometric material on exactness, Hodge-like decompositions, and Weyl chambers serves as a structural language for interpreting the objects developed herein, rather than as a claim that a fully new theorem system has been closed within this manuscript.

\paragraph{Roadmap.}
Sections~2--3 build the exact pairwise-margin core for scores and linear models. Section~4 contains the canonical local-share theorem and the exchange-geometry theorem for global influence structure and Influence Exchange. Sections~5--6 lift the same objects to permutation space and score-space geometry. Section~7 extends the framework to nonlinear scoring and proves the interaction-curvature characterization of factorwise path dependence. Section~8 reinterprets this nonlinear obstruction in discrete exactness, cycle-consistency, and triangle-curl language; Section~9 closes the geometric arc in root-space terms. Related work is deferred to Section~10, followed by limitations and conclusion.

\section{Preliminaries: Scores, Orders, and Pairwise Margins}

This section establishes the basic objects of the paper. The motivating observation is elementary: a score vector is merely a representation of a ranking, whereas pairwise margins are the atomic coordinates that determine order.

\begin{assumption}[No ties in the main text]
Unless stated otherwise, we assume \(s_i \neq s_j\) for all distinct items \(i \neq j\). Ties and degenerate pairs are discussed in Appendix~C and Section~11.
\end{assumption}

\begin{definition}[Score-induced ranking]
Given \(n\) items with score vector \(s = (s_1,\dots,s_n) \in \R^n\), the induced ranking is the total order
\[
    i \succ_s j \quad \Longleftrightarrow \quad s_i > s_j.
\]
\end{definition}

\subsection{Monotone transforms preserve ranking order}

\begin{proposition}[Strictly monotone invariance]
Let \(g:\R \to \R\) be strictly increasing. Then the ranking induced by \(s\) is identical to the ranking induced by \(g(s) = (g(s_1),\dots,g(s_n))\).
\end{proposition}

\begin{proof}
For any pair \(i,j\), strict monotonicity implies
\[
    s_i > s_j \quad \Longleftrightarrow \quad g(s_i) > g(s_j).
\]
Hence every pairwise comparison is preserved, and therefore the total order is preserved.
\end{proof}

Although elementary, Proposition~2.1 clarifies the paper's perspective. Absolute score values carry no intrinsic ranking meaning; they matter only insofar as they encode pairwise order. This motivates treating pairwise quantities as the primary analytical objects throughout.

\subsection{Pairwise margins and atomic flips}

\begin{definition}[Pairwise margin]
For distinct items \(i\) and \(j\), the pairwise margin is
\begin{equation}
    \Delta_{ij} = s_i - s_j.
\end{equation}
\end{definition}

\begin{proposition}[Margin sign determines pairwise order]
Under the no-ties assumption,
\[
    i \succ_s j \quad \Longleftrightarrow \quad \Delta_{ij} > 0.
\]
Equivalently, \(j \succ_s i\) if and only if \(\Delta_{ij} < 0\).
\end{proposition}

\begin{proof}
The claim is immediate from the definition \(\Delta_{ij} = s_i - s_j\).
\end{proof}

\begin{remark}[Pairwise order determines the total order]
Under no ties, the collection of signs \(\sign(\Delta_{ij})\) for all unordered pairs \(\{i,j\}\) determines the induced total order. In that sense, the ranking can be reconstructed from pairwise margins even though it is represented by a score vector.
\end{remark}

\begin{remark}[Atomic flips as sign crossings]
Consider a continuous perturbation path \(t \mapsto s(t)\) in score space. A local flip between items \(i\) and \(j\) occurs exactly when the pairwise margin
\[
    \Delta_{ij}(t) = s_i(t) - s_j(t)
\]
crosses zero. Thus, at the most local level, a ranking change is a sign crossing of a pairwise margin.
\end{remark}

This observation serves as the entry point for the remainder of the manuscript. Once ranking changes are analyzed at the level of pairwise margins, linear factor decompositions, local shares, and global influence accounting arise naturally rather than as ad hoc constructions. Section~3 develops the exact decomposition for the linear case.

\section{Linear Margin Decomposition and Pairwise Attribution}

The linear case constitutes the formal foundation of the paper. It is introduced not because real systems are typically linear, but because every object of interest can be defined exactly and transparently in this setting.

\subsection{Exact factor-level decomposition}

Assume a linear scoring model
\begin{equation}
    s_i = w^\top x_i = \sum_{f=1}^d w_f x_{if}.
\end{equation}

\begin{definition}[Linear factor-level local contribution]
For a pair \((i,j)\) and factor \(f\), define the local pairwise contribution
\begin{equation}
    \Delta_{ij}^{(f)} = w_f(x_{if} - x_{jf}).
\end{equation}
\end{definition}

\begin{proposition}[Linear margin decomposition]
Under the linear model,
\begin{equation}
    \Delta_{ij} = \sum_{f=1}^d \Delta_{ij}^{(f)}.
\end{equation}
\end{proposition}

\begin{proof}
\[
    \Delta_{ij}
    = s_i - s_j
    = \sum_{f=1}^d w_f x_{if} - \sum_{f=1}^d w_f x_{jf}
    = \sum_{f=1}^d w_f(x_{if} - x_{jf})
    = \sum_{f=1}^d \Delta_{ij}^{(f)}.
\]
\end{proof}

The proposition is algebraically immediate yet conceptually central. The pairwise margin is not an opaque aggregate of the scoring model; in the linear case, it decomposes as an exact sum of factor-level local terms, each supporting one side of the pairwise contest.

\begin{example}[Two-factor local contest]
Consider two factors with local contributions \(\Delta_{ij}^{(1)} = 3\) and \(\Delta_{ij}^{(2)} = -2\). Then the total margin is \(\Delta_{ij} = 1\), so item \(i\) outranks item \(j\). Factor \(1\) supports item \(i\), factor \(2\) supports item \(j\), and the net order is the algebraic result of their opposition.
\end{example}

\subsection{From score attribution to pairwise attribution}

When ranking is analyzed through pairwise order, the natural explanatory question shifts. Rather than asking only why a given item has a high absolute score, we ask which factors drove the margin in a specific pairwise contest and which factors would be responsible if that contest were to flip.

This shift suggests four design criteria for a local pairwise attribution rule.
\begin{enumerate}[leftmargin=2em]
    \item \textbf{Locality.} The attribution for a pair \((i,j)\) should depend on that pair rather than on unrelated items.
    \item \textbf{Nonnegativity at the share level.} If the goal is to allocate local influence budget, shares should not cancel by sign.
    \item \textbf{Directional separation.} A factor that supports item \(i\) and a factor that supports item \(j\) may oppose each other, yet both may exert nontrivial local influence.
    \item \textbf{Budget normalization.} The local accounting should admit a normalized allocation that sums to one on informative pairs.
\end{enumerate}

These criteria motivate converting signed linear terms into nonnegative local effort. The following section introduces a refinement-consistency requirement and demonstrates that, for local budgeting of such effort, the influence-share normalization is uniquely determined.

\section{Influence Share, Influence Structure, and Influence Exchange}

This section turns the exact linear decomposition into an influence-accounting framework. The development proceeds from signed local contributions to a canonical nonnegative local budget, from local shares to a global influence structure, and finally from global conservation to the geometry and rigidity of Influence Exchange.

\subsection{Local influence share}

\begin{definition}[Informative-pair effort]
For a pair \((i,j)\) in the linear model, define the total local effort
\begin{equation}
    Z_{ij} = \sum_{f=1}^d \left|\Delta_{ij}^{(f)}\right|.
\end{equation}
\end{definition}

\begin{definition}[Influence share]
On an informative pair with \(Z_{ij} > 0\), define the local influence share of factor \(f\) by
\begin{equation}
    \rho_{ij}^{(f)} = \frac{|\Delta_{ij}^{(f)}|}{Z_{ij}}.
\end{equation}
\end{definition}

The use of absolute values is intentional: the objective is to quantify how much local pairwise control a factor exerts, not merely the sign of its preferred direction. This definition yields a nonnegative local budget allocation, but leaves open a structural question: is this normalization merely convenient, or is it forced by a minimal invariance principle? The following theorem shows that once local budgeting and stability under pure factor refinement are imposed, the normalization is uniquely determined.

\begin{theorem}[Refinement-consistency characterization of local influence share]
\label{thm:refinement-local-share}
For each dimension \(d\), let
\[
    \Psi^{(d)}:\R_{\ge 0}^d\setminus\{0\}\to \Delta^{d-1}
\]
be a local share rule for nonnegative effort vectors. Assume the family \(\{\Psi^{(d)}\}_{d\ge 1}\) satisfies budgeting and refinement consistency: whenever the \(k\)-th effort \(a_k\) is split into \(r,q\ge 0\) with \(r+q=a_k\), the refined vector
\[
    a^{[k\to(r,q)]}
    =
    (a_1,\dots,a_{k-1},r,q,a_{k+1},\dots,a_d)
\]
satisfies, for \(a_k>0\),
\[
    \Psi^{(d+1)}\!\left(a^{[k\to(r,q)]}\right)
    =
    \left(
    \Psi_1^{(d)}(a),\dots,\Psi_{k-1}^{(d)}(a),
    \frac{r}{a_k}\Psi_k^{(d)}(a),
    \frac{q}{a_k}\Psi_k^{(d)}(a),
    \Psi_{k+1}^{(d)}(a),\dots,\Psi_d^{(d)}(a)
    \right),
\]
and if \(a_k=0\), the same identity is interpreted as leaving all non-refined coordinates unchanged and assigning share \(0\) to the two refined coordinates. Then the only possible rule is
\begin{equation}
    \Psi_f^{(d)}(a)=\frac{a_f}{\sum_{g=1}^d a_g}.
\end{equation}
\end{theorem}

\begin{proof}
Let \(S=\sum_{g=1}^d a_g>0\). For the one-dimensional effort vector \((S)\), budgeting forces \(\Psi^{(1)}(S)=(1)\). Now refine this single coordinate repeatedly:
\[
    (S)\to (a_1,S-a_1)\to (a_1,a_2,S-a_1-a_2)\to \cdots \to (a_1,\dots,a_d).
\]
At each nonzero split, refinement consistency forces the parent's share to be divided in the same ratio as its effort. Zero-effort children receive zero share by assumption. Therefore the share assigned to the final coordinate \(f\) is \(a_f/S\). This proves uniqueness. Conversely, the formula \(a_f/\sum_g a_g\) is nonnegative, sums to one, and obeys the displayed refinement identity directly, so it exists.
\end{proof}

\begin{corollary}[Canonical linear local share]
\label{cor:canonical-linear-share}
In the linear core of this paper, applying Theorem~\ref{thm:refinement-local-share} to the effort vector \(a_f=|\Delta_{ij}^{(f)}|\) yields Definition~4.2. Thus \(\rho_{ij}^{(f)}\) is the unique refinement-consistent local budgeting rule for linear pairwise effort.
\end{corollary}

\begin{remark}[Bookkeeping refinements]
If a factor is split only as a bookkeeping refinement of its nonnegative effort, the refined shares inside that block sum to the original share. Consequently, block-level global influence and block-level exchange are unchanged by such a pure refinement.
\end{remark}

\begin{proposition}[Simplex property of local shares]
For every informative pair \((i,j)\), the vector \((\rho_{ij}^{(1)},\dots,\rho_{ij}^{(d)})\) is nonnegative and sums to one.
\end{proposition}

\begin{proof}
Nonnegativity is immediate from the absolute value. Also,
\[
    \sum_{f=1}^d \rho_{ij}^{(f)}
    = \sum_{f=1}^d \frac{|\Delta_{ij}^{(f)}|}{Z_{ij}}
    = \frac{\sum_{f=1}^d |\Delta_{ij}^{(f)}|}{Z_{ij}}
    = 1.
\]
\end{proof}

\subsection{Global influence structure}

\begin{definition}[Pair distribution]
Let \(\Pcal\) be a probability distribution on ordered or unordered item pairs. The choice of \(\Pcal\) is external to the framework and can encode, for example, exposure weighting, top-of-list emphasis, or a uniform diagnostic view.
Unless an explicit convention is supplied for degenerate pairs, the global quantities below are evaluated with \(\Pcal\) supported on informative pairs for which \(Z_{ij}>0\).
\end{definition}

\begin{definition}[Global influence structure]
The global influence structure is the vector
\begin{equation}
    I_f = \E_{(i,j)\sim\Pcal}\left[\rho_{ij}^{(f)}\right], \qquad
    \Ivec(\theta) = \left(I_f(\theta)\right)_{f=1}^d,
\end{equation}
where \(\theta\) denotes the model state.
\end{definition}

\begin{proposition}[Global conservation]
The global influence structure satisfies
\begin{equation}
    \sum_{f=1}^d I_f = 1.
\end{equation}
\end{proposition}

\begin{proof}
By linearity of expectation and the simplex property,
\[
    \sum_{f=1}^d I_f
    = \sum_{f=1}^d \E_{(i,j)\sim\Pcal}\left[\rho_{ij}^{(f)}\right]
    = \E_{(i,j)\sim\Pcal}\left[\sum_{f=1}^d \rho_{ij}^{(f)}\right]
    = \E_{(i,j)\sim\Pcal}[1]
    = 1.
\]
\end{proof}

The global structure therefore remains a normalized allocation. It admits interpretation as a system-level summary of which factors, on average under \(\Pcal\), exercise what proportion of pairwise control.

The local budgeting rule is now canonically determined. The natural question is whether the resulting global influence structure possesses an intrinsic geometric structure.

\subsection{Influence Exchange}

\begin{definition}[Influence Exchange]
Given two model states \(\theta\) and \(\theta'\), define the \emph{Influence Exchange} vector by
\begin{equation}
    \Exch(\theta,\theta') = \Ivec(\theta') - \Ivec(\theta).
\end{equation}
In coordinates, \(\Exch_f(\theta,\theta') = I_f(\theta') - I_f(\theta)\).
\end{definition}

\begin{corollary}[Zero-sum global exchange]
For any two model states \(\theta,\theta'\),
\begin{equation}
    \sum_{f=1}^d \Exch_f(\theta,\theta') = 0.
\end{equation}
\end{corollary}

\begin{proof}
The claim follows immediately from global conservation:
\[
    \sum_{f=1}^d \Exch_f(\theta,\theta')
    = \sum_{f=1}^d I_f(\theta') - \sum_{f=1}^d I_f(\theta)
    = 1 - 1
    = 0.
\]
\end{proof}

The zero-sum identity is a conservation consequence. A deeper question concerns how the global influence structure is generated, how it responds when parameters change, and when it can remain entirely invariant under such changes. In the linear model, after reparameterizing by log-absolute weights, the answer takes a particularly clean form.

\subsection{Exchange geometry, Laplacian response, and zero-exchange rigidity}

For this subsection write \(p=(i,j)\) and
\[
    b_p^{(f)}=|x_{if}-x_{jf}|.
\]
Let
\[
    \mathcal F_*=\left\{f: |w_f|>0 \text{ and } \mathbb P_{p\sim\Pcal}(b_p^{(f)}>0)>0\right\}
\]
be the active nonzero factor set, and assume \(\Pcal\) is supported on pairs with \(\sum_{k\in\mathcal F_*} b_p^{(k)}>0\). Factors outside \(\mathcal F_*\) contribute only trivial zero share in the log-parameterized analysis below.

\begin{theorem}[Exchange geometry, Laplacian response, and zero-exchange rigidity]
\label{thm:exchange-geometry}
For \(f\in\mathcal F_*\), set \(u_f=\log |w_f|\) and
\begin{equation}
    \rho_p^{(f)}(u)
    =
    \frac{e^{u_f} b_p^{(f)}}{\sum_{k\in\mathcal F_*} e^{u_k}b_p^{(k)}}.
\end{equation}
Define
\[
    I_f(u)=\E_{p\sim\Pcal}[\rho_p^{(f)}(u)],
    \qquad
    I(u)=(I_f(u))_{f\in\mathcal F_*},
\]
\[
    \Phi(u)=
    \E_{p\sim\Pcal}\left[
    \log\!\left(\sum_{k\in\mathcal F_*} e^{u_k}b_p^{(k)}\right)
    \right],
\]
and the factor competition graph \(G_{\Pcal}=(\mathcal F_*,E_{\Pcal})\) by
\[
    \{f,g\}\in E_{\Pcal}
    \quad\Longleftrightarrow\quad
    \mathbb P_{p\sim\Pcal}(b_p^{(f)}b_p^{(g)}>0)>0.
\]
Then:
\begin{enumerate}[leftmargin=2em]
    \item \textbf{Gradient representation.} \(I(u)=\nabla\Phi(u)\).
    \item \textbf{Laplacian Jacobian.} The Jacobian \(J(u)=\nabla^2\Phi(u)\) is
    \begin{equation}
        J(u)=
        \E\!\left[\operatorname{Diag}(\rho_p(u))-\rho_p(u)\rho_p(u)^\top\right].
    \end{equation}
    In particular, for \(f\neq g\),
    \[
        J_{fg}(u)=-\E[\rho_p^{(f)}(u)\rho_p^{(g)}(u)]\le 0,
        \qquad
        J_{ff}(u)=\sum_{g\neq f}\E[\rho_p^{(f)}(u)\rho_p^{(g)}(u)].
    \]
    Hence \(J(u)\) is a symmetric positive-semidefinite Laplacian-type matrix supported on \(G_{\Pcal}\).
    \item \textbf{Finite exchange energy identity.} For any \(u,u'\), let \(\Delta u=u'-u\). Then
    \begin{equation}
        \left\langle \Delta u, I(u')-I(u)\right\rangle
        =
        \int_0^1 \Delta u^\top J(u+t\Delta u)\Delta u\,dt
        \ge 0.
    \end{equation}
    \item \textbf{Zero-exchange rigidity.}
    \[
        I(u')=I(u)
        \quad\Longleftrightarrow\quad
        \Delta u \text{ is constant on every connected component of }G_{\Pcal}.
    \]
    In particular, if \(G_{\Pcal}\) is connected, then \(I(u')=I(u)\) if and only if \(u'=u+c\one\) for some scalar \(c\).
\end{enumerate}
\end{theorem}

\begin{proof}
For one pair, define
\[
    \phi_p(u)=\log\!\left(\sum_{k\in\mathcal F_*} e^{u_k}b_p^{(k)}\right).
\]
Direct differentiation gives \(\partial_f\phi_p(u)=\rho_p^{(f)}(u)\). Taking expectation gives \(I(u)=\nabla\Phi(u)\). Differentiating once more gives
\[
    \partial_{fg}^2\phi_p(u)=\rho_p^{(f)}(u)\left(\mathbf{1}_{f=g}-\rho_p^{(g)}(u)\right),
\]
and hence the displayed Hessian formula after expectation.

For any vector \(v\),
\[
    v^\top J(u)v
    =
    \frac12\,\E_p\left[
    \sum_{f,g\in\mathcal F_*}
    \rho_p^{(f)}(u)\rho_p^{(g)}(u)(v_f-v_g)^2
    \right]\ge 0.
\]
This is the Laplacian quadratic form: off-diagonal entries are nonpositive co-activation weights and rows sum to zero.

The finite identity follows from the fundamental theorem of calculus along the segment \(u+t\Delta u\):
\[
    I(u')-I(u)=\int_0^1 J(u+t\Delta u)\Delta u\,dt.
\]
Taking the inner product with \(\Delta u\) yields the energy identity.

It remains to identify the zero-energy directions. The quadratic form above vanishes exactly when \(v_f=v_g\) for every pair of factors that co-occur with positive \(\Pcal\)-probability, and therefore exactly when \(v\) is constant on connected components of \(G_{\Pcal}\). If \(\Delta u\) has that form, every pairwise softmax \(\rho_p(u)\) is unchanged, so \(I(u')=I(u)\). Conversely, if \(I(u')=I(u)\), the energy identity has a nonnegative continuous integrand and value zero. Thus its integrand is identically zero, forcing \(\Delta u\) to lie in the same componentwise-constant nullspace. This proves the rigidity statement.
\end{proof}

Theorem~\ref{thm:exchange-geometry} elevates Influence Exchange from a difference vector to a response object generated by a convex potential. The competition graph encodes which factors are jointly priced in local pairwise contests; when it is connected, global influence identifies relative absolute weights up to the unavoidable common-scaling gauge \(u\mapsto u+c\one\).

The central chain of the paper can now be written compactly as
\begin{equation}
    s_i \;\longrightarrow\; \Delta_{ij}^{(f)} \;\longrightarrow\; \rho_{ij}^{(f)} \;\longrightarrow\; I_f \;\longrightarrow\; \Exch_f.
\end{equation}
Influence Exchange is therefore not an auxiliary metric appended post hoc. It is the model-to-model difference object induced by the complete pairwise accounting pipeline depicted in Figure~\ref{fig:pipeline}, and in the linear core it is constrained by the exchange geometry above.

\section{Ranking Changes as Transpositions in Permutation Space}

The preceding section described the continuous redistribution of pairwise control: local shares aggregate into a global influence structure, and log-absolute linear weights move that structure through a convex-gradient geometry. This section studies a different object: the discrete outcome change after margins cross zero. The goal is not to replace pairwise accounting with permutation theory, but to show that the local flip picture has a precise counterpart in permutation space.

\subsection{Rankings as permutations}

Under the no-ties assumption, a ranking of \(n\) items can be identified with a permutation \(\sigma \in S_n\). We use the position-map convention: \(\sigma(i)\) is the position of item \(i\), and smaller values mean higher rank.

A change from one ranking to another is therefore a change from one permutation to another. Standard results in permutation theory imply that any such change can be decomposed into transpositions, and that adjacent transpositions generate the whole group \(S_n\) \citep{bjornerbrenti2005}.

\begin{figure}[t]
    \centering
    \includegraphics[width=0.72\linewidth]{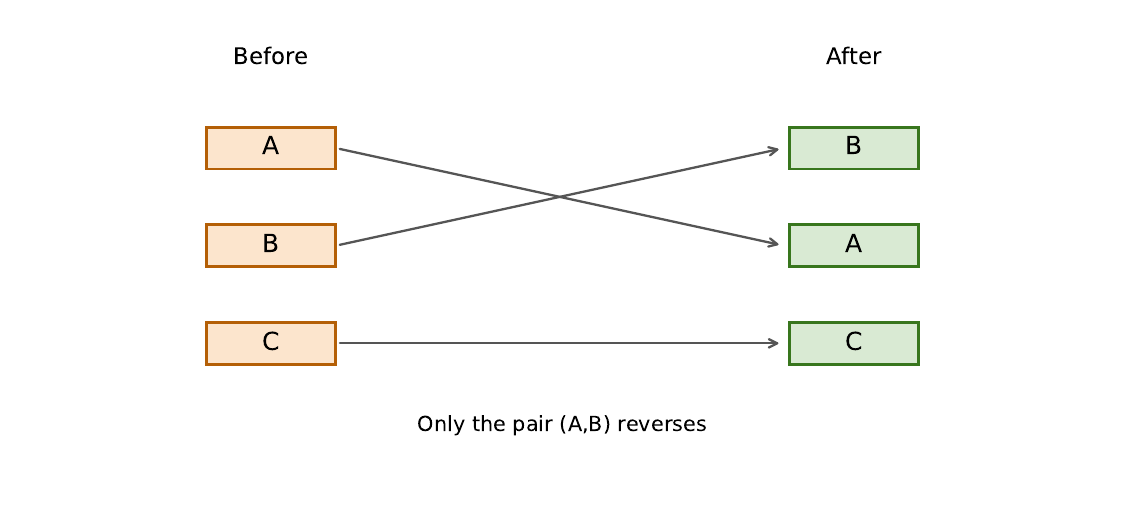}
    \caption{A three-item change from \(A \succ B \succ C\) to \(B \succ A \succ C\). Only the pair \((A,B)\) reverses. This is the discrete counterpart of a single pairwise margin crossing zero.}
    \label{fig:atomic-flip}
\end{figure}

\begin{example}[A single local flip]
Suppose the initial ranking is \(A \succ B \succ C\) and the updated ranking is \(B \succ A \succ C\). At the level of pairwise relations, only the order of \(A\) and \(B\) changes. In permutation language this is one adjacent transposition; in margin language it is one pairwise sign crossing.
\end{example}

\subsection{Kendall distance counts flips}

The Kendall tau distance between two rankings \(\sigma\) and \(\pi\) is
\begin{equation}
    \KT(\sigma,\pi)
    =
    \left|
    \left\{
    (i,j): i<j,\,
    (\sigma(i)-\sigma(j))(\pi(i)-\pi(j)) < 0
    \right\}
    \right|.
\end{equation}
It equals the number of discordant pairs and, equivalently, the minimum number of adjacent transpositions needed to transform one ranking into the other \citep{kendall1938new,bjornerbrenti2005}.

Kendall tau distance thus provides a natural measure of pairwise disorder; however, it remains purely a counting measure: it measures how many local pairwise relations changed, not which factors drove those changes. Influence Exchange is designed precisely to address this missing attribution layer. In other words, Kendall tau enumerates flips; the exchange-geometry theorem describes the global repricing and reallocation of pairwise control beneath them.

\subsection{Bridge back to geometry}

The permutation-theoretic language clarifies the discreteness of ranking outcomes. The following section complements it with a continuous view: score vectors move in \(\R^n\), pairwise boundaries appear as hyperplanes, and a discrete flip corresponds to a geometric boundary crossing.

\section{Hyperplane Geometry and Margin as a Normal Coordinate}

Permutation space describes the discrete outputs of a ranking system, while score space describes the continuous domain in which those outputs are induced. This section bridges these two perspectives and provides a geometric interpretation of pairwise margins and factor-level contributions.

\subsection{Score-space chambers}

Let \(s=(s_1,\dots,s_n)\in\R^n\). For each pair \(i \neq j\), define the boundary hyperplane
\begin{equation}
    H_{ij} = \{ s \in \R^n : s_i = s_j \}.
\end{equation}
The collection of hyperplanes \(\{H_{ij}\}_{i<j}\) is the braid arrangement of type \(A_{n-1}\) \citep{stanley2007hyperplane,humphreys1990reflection}. Its chambers correspond to strict total orders of the coordinates, hence to rankings of the \(n\) items.

\begin{figure}[t]
    \centering
    \includegraphics[width=0.74\linewidth]{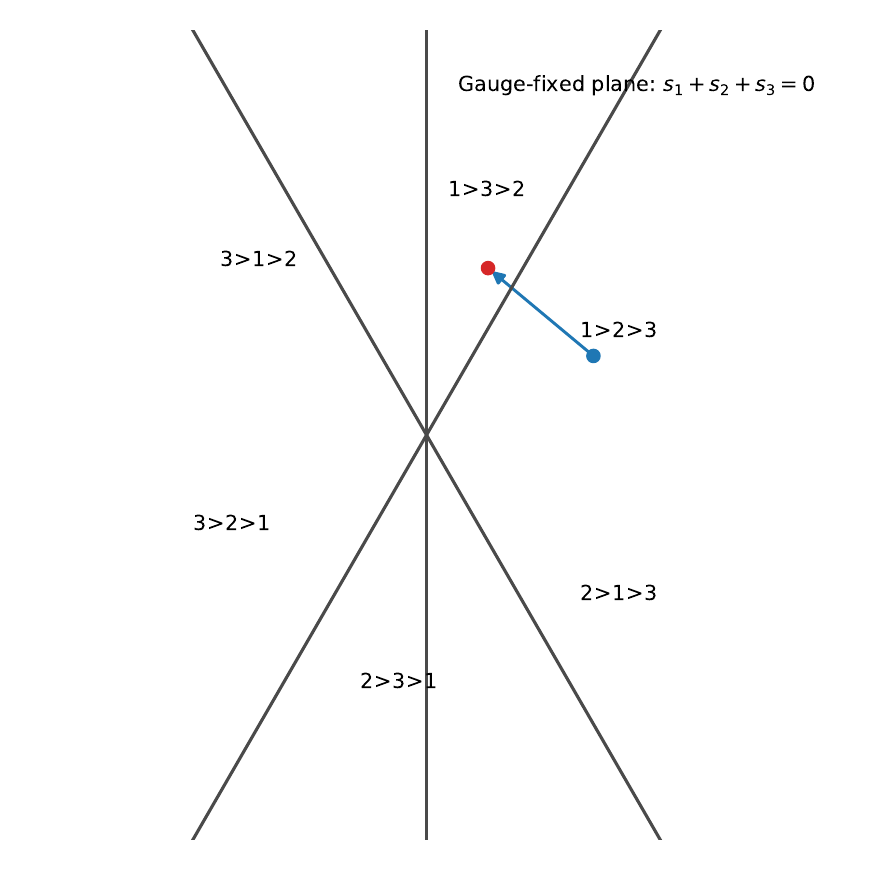}
    \caption{The \(n=3\) chamber picture in the gauge-fixed plane \(s_1+s_2+s_3=0\). Each chamber corresponds to one permutation. Crossing a boundary \(s_i=s_j\) flips one pairwise relation.}
    \label{fig:hyperplane-chambers}
\end{figure}

Thus, a ranking change corresponds geometrically to a chamber transition. If a perturbation path crosses \(H_{ij}\), then the order of \(i\) and \(j\) reverses. Figure~\ref{fig:hyperplane-chambers} illustrates this for the three-item case.

\subsection{Pairwise margin as a normal coordinate}

Let \(\e_i\) denote the \(i\)-th standard basis vector of \(\R^n\). The hyperplane \(H_{ij}\) has normal vector \(\e_i-\e_j\).

\begin{proposition}[Margin as normal coordinate]
For every score vector \(s \in \R^n\),
\begin{equation}
    \langle s, \e_i-\e_j \rangle = s_i - s_j = \Delta_{ij}.
\end{equation}
\end{proposition}

\begin{proof}
By direct calculation,
\[
    \langle s, \e_i-\e_j \rangle
    = \langle s,\e_i\rangle - \langle s,\e_j\rangle
    = s_i - s_j
    = \Delta_{ij}.
\]
\end{proof}

\begin{remark}[Distance to the boundary]
Since \(\|\e_i-\e_j\| = \sqrt{2}\), the Euclidean distance from \(s\) to \(H_{ij}\) is
\begin{equation}
    \dist(s,H_{ij}) = \frac{|s_i-s_j|}{\sqrt{2}} = \frac{|\Delta_{ij}|}{\sqrt{2}}.
\end{equation}
Hence the magnitude of the pairwise margin is a scaled distance-to-boundary quantity.
\end{remark}

\subsection{Factor slices of the normal direction}

In the linear case, the decomposition
\[
    \Delta_{ij} = \sum_{f=1}^d \Delta_{ij}^{(f)}
\]
admits a geometric interpretation: each factor contributes a component of the total normal coordinate along the direction \(\e_i-\e_j\). The signed quantities \(\Delta_{ij}^{(f)}\) determine which side of the boundary each factor supports, while the magnitudes \(|\Delta_{ij}^{(f)}|\) quantify the extent of that factor's participation in the local contest. The influence share thus represents the normalized magnitude of factor-level participation along the relevant normal direction.

\subsection{Toward nonlinear geometry}

The hyperplane picture is exact in score space regardless of how the scores were produced. What changes in nonlinear models is the structure prior to score computation: the preimage of the boundary \(s_i=s_j\) in feature space need not be linear. Section~7 accordingly retains the pairwise boundary viewpoint while pulling it back to feature space and replacing global linear decomposition with path-based pairwise accounting.

\section{Nonlinear Extensions via Local Linearization and Path-Based Pairwise Attribution}

The linear case yields exact factor-level accounting as an immediate consequence of linearity. Nonlinearity does not invalidate the pairwise margin itself; rather, it prevents a unique factorwise decomposition, as cross-factor interactions introduce path dependence. This section extends the framework, proves the exact no-interaction boundary for factorwise pathwise uniqueness, and maintains a clear separation between formal statements and semantic choices.

\subsection{Pairwise invariance and factorwise ambiguity}

Let the score of item \(i\) be \(s_i = F(x_i)\), where \(F:\R^d \to \R\) is differentiable. The pairwise margin remains
\begin{equation}
    \Delta_{ij} = F(x_i) - F(x_j),
\end{equation}
so the ranking decision still depends only on its sign. The difference is that a factor-level decomposition of \(\Delta_{ij}\) is no longer unique.

\begin{example}[A bilinear interaction term]
Take \(F(x_1,x_2) = x_1x_2\). For two items \(i\) and \(j\),
\begin{align}
    \Delta_{ij}
    &= x_{i1}x_{i2} - x_{j1}x_{j2} \nonumber \\
    &= x_{j2}(x_{i1}-x_{j1}) + x_{i1}(x_{i2}-x_{j2}) \nonumber \\
    &= x_{i2}(x_{i1}-x_{j1}) + x_{j1}(x_{i2}-x_{j2}) \nonumber \\
    &= \frac{x_{i2}+x_{j2}}{2}(x_{i1}-x_{j1}) + \frac{x_{i1}+x_{j1}}{2}(x_{i2}-x_{j2}). \label{eq:bilinear-three-decompositions}
\end{align}
All three decompositions are valid, yet they assign different factor-level shares. The pairwise margin is unique; the factor-level accounting is not.
\end{example}

This illustrates a central observation of the paper: the pairwise ranking problem remains invariant, while cross-factor interaction renders the factor-level accounting path-dependent.

\subsection{Midpoint local linearization}

To recover local structure, define
\[
    \barx = \frac{x_i+x_j}{2},
    \qquad
    \deltax = x_i - x_j.
\]

\begin{proposition}[Midpoint local linearization]
If \(F\) is \(C^3\) in a neighborhood of \(\barx\), then
\begin{equation}
    \Delta_{ij}
    =
    F\!\left(\barx + \frac{\deltax}{2}\right)
    -
    F\!\left(\barx - \frac{\deltax}{2}\right)
    =
    \nabla F(\barx)^\top \deltax
    + O(\|\deltax\|^3).
\end{equation}
\end{proposition}

The proof, which proceeds by symmetric Taylor expansion, is given in Appendix~A. The essential observation is that even-order terms cancel in this midpoint expansion. Consequently, the leading control term is
\begin{equation}
    \Delta_{ij} \approx \sum_{f=1}^d \partial_f F(\barx)\,(x_{if}-x_{jf}).
\end{equation}

\begin{remark}[Local pricing]
The coefficients \(\partial_f F(\barx)\) act as local prices for feature differences around the pair midpoint. Unlike linear weights, these prices depend on the pair and on location. We use the term \emph{pricing} only as an interpretation of local sensitivity, not as an economic claim.
\end{remark}

\subsection{Path-based pairwise attribution}

Local linearization is informative but inherently approximate. For an exact decomposition of the nonlinear pairwise margin, let \(\gamma:[0,1]\to \R^d\) be a differentiable path from \(x_j\) to \(x_i\), that is, \(\gamma(0)=x_j\) and \(\gamma(1)=x_i\).

\begin{definition}[Path-based pairwise attribution]
Define the factor-level path contribution
\begin{equation}
    \Delta_{ij}^{\gamma,(f)}
    =
    \int_0^1 \partial_f F(\gamma(t))\,\dot{\gamma}_f(t)\,dt.
\end{equation}
\end{definition}

\begin{proposition}[Completeness of path-based pairwise attribution]
For any differentiable path \(\gamma\) from \(x_j\) to \(x_i\),
\begin{equation}
    \sum_{f=1}^d \Delta_{ij}^{\gamma,(f)} = \Delta_{ij}.
\end{equation}
\end{proposition}

\begin{proof}
By the chain rule,
\[
    \frac{d}{dt}F(\gamma(t)) = \nabla F(\gamma(t))^\top \dot{\gamma}(t)
    = \sum_{f=1}^d \partial_f F(\gamma(t))\,\dot{\gamma}_f(t).
\]
Integrating from \(0\) to \(1\) and applying the fundamental theorem of calculus yields
\[
    F(x_i)-F(x_j) = \int_0^1 \sum_{f=1}^d \partial_f F(\gamma(t))\,\dot{\gamma}_f(t)\,dt
    = \sum_{f=1}^d \Delta_{ij}^{\gamma,(f)}.
\]
\end{proof}

Completeness holds exactly, but the componentwise allocation is generally path-dependent. The total line integral is path-independent because it derives from the scalar potential \(F\); the individual factor-level components need not be.

\subsection{Pairwise Integrated Gradients}

The most natural path is the straight line
\begin{equation}
    \gamma_{\mathrm{SL}}(t) = x_j + t(x_i-x_j).
\end{equation}
It yields the pairwise analogue of Integrated Gradients \citep{sundararajan2017axiomatic,aumann1974values}:
\begin{equation}
    \Delta_{ij}^{\mathrm{PIG},(f)}
    =
    (x_{if}-x_{jf})
    \int_0^1 \partial_f F(x_j+t(x_i-x_j))\,dt.
\end{equation}

\begin{proposition}[Linear recovery]
If \(F(x)=w^\top x\) is linear, then for every path and every factor \(f\),
\[
    \Delta_{ij}^{\gamma,(f)} = w_f(x_{if}-x_{jf}) = \Delta_{ij}^{(f)}.
\]
In particular, Pairwise Integrated Gradients recovers the linear decomposition exactly.
\end{proposition}

\begin{proof}
For a linear function, \(\partial_f F \equiv w_f\) is constant. Hence
\[
    \Delta_{ij}^{\gamma,(f)}
    = \int_0^1 w_f\,\dot{\gamma}_f(t)\,dt
    = w_f(\gamma_f(1)-\gamma_f(0))
    = w_f(x_{if}-x_{jf}).
\]
\end{proof}

Straight-line path attribution is therefore not an extraneous construction; it constitutes a direct generalization of the linear factor decomposition.

\begin{figure}[t]
    \centering
    \includegraphics[width=0.76\linewidth]{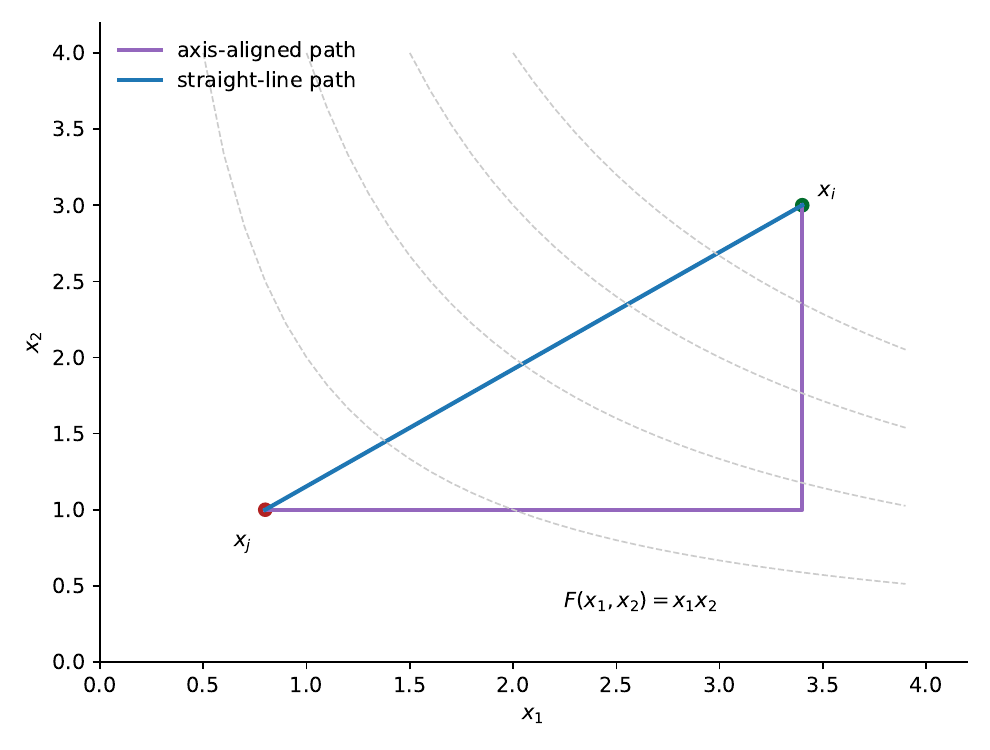}
    \caption{Two paths between the same pair of feature vectors in a bilinear model. The total margin change is identical, but the factor-level decomposition depends on the path semantics.}
    \label{fig:nonlinear-paths}
\end{figure}

\subsection{Interaction curvature and factorwise path dependence}

The difficulty in nonlinear models lies not with the pairwise margin itself, which remains exact, but with the uniqueness of the factor-level decomposition used to account for that margin. Differential-form notation gives a compact way to state the obstruction: the factorwise contribution associated with coordinate \(f\) is the line integral of the \(1\)-form \(\omega_f=\partial_fF(x)\,dx_f\). We use standard differential-form facts such as Stokes' theorem and the Poincar\'e lemma \citep{lee2012smooth}. The next theorem identifies the precise obstruction to path-independent factorwise accounting.

\begin{theorem}[Interaction-curvature characterization of factorwise path dependence]
\label{thm:interaction-curvature}
Let \(\Omega=\Omega_1\times\cdots\times\Omega_d\subset\R^d\) be an open rectangle, with each \(\Omega_f\subset\R\) an open interval, and let \(F\in C^2(\Omega)\). For each factor \(f\), define
\[
    \omega_f=\partial_f F(x)\,dx_f.
\]
For every piecewise \(C^1\) path \(\gamma\) from \(x_j\) to \(x_i\) in \(\Omega\), define
\[
    \Delta_{ij}^{\gamma,(f)}=\int_\gamma \omega_f.
\]
Then:
\begin{enumerate}[leftmargin=2em]
    \item \textbf{Continuous curl.}
    \begin{equation}
        d\omega_f=
        \sum_{g\neq f}\partial_{gf}F(x)\,dx_g\wedge dx_f.
    \end{equation}
    \item \textbf{Surface-flux formula.} If two paths \(\gamma_0,\gamma_1\) with the same endpoints bound an oriented piecewise smooth surface \(\Sigma\), then
    \begin{equation}
        \Delta_{ij}^{\gamma_1,(f)}-\Delta_{ij}^{\gamma_0,(f)}
        =
        \oint_{\gamma_1-\gamma_0}\omega_f
        =
        \int_\Sigma d\omega_f
        =
        \sum_{g\neq f}\int_\Sigma \partial_{gf}F(x)\,dx_g\wedge dx_f.
    \end{equation}
    \item \textbf{Fixed-factor equivalence.} For a fixed factor \(f\), the following are equivalent: \(\Delta_{ij}^{\gamma,(f)}\) is path-independent for all endpoints and paths; \(\partial_{gf}F\equiv 0\) for every \(g\neq f\); and \(\omega_f\) is the differential of a one-variable primitive \(\Phi_f(x_f)\), equivalently \(\partial_fF\) depends only on \(x_f\).
    \item \textbf{Global equivalence.} The following are equivalent: every factor attribution \(\Delta_{ij}^{\gamma,(f)}\) is path-independent; all mixed partials \(\partial_{gf}F\) vanish identically for \(g\neq f\); and
    \begin{equation}
        F(x)=c+\sum_{f=1}^d \Phi_f(x_f)
    \end{equation}
    on \(\Omega\), for some one-variable potentials \(\Phi_f:\Omega_f\to\R\).
\end{enumerate}
\end{theorem}

\begin{proof}
Because \(d(\partial_fF)=\sum_g\partial_{gf}F\,dx_g\), we have
\[
    d\omega_f
    =
    \sum_g \partial_{gf}F\,dx_g\wedge dx_f.
\]
The \(g=f\) term vanishes since \(dx_f\wedge dx_f=0\), proving the continuous-curl formula.

For two paths with the same endpoints, the oriented loop \(\gamma_1-\gamma_0\) has boundary \(\partial\Sigma\). Stokes' theorem gives
\[
    \int_{\gamma_1}\omega_f-\int_{\gamma_0}\omega_f
    =
    \oint_{\gamma_1-\gamma_0}\omega_f
    =
    \int_\Sigma d\omega_f,
\]
and substituting the first formula gives the displayed flux identity.

On the open rectangle \(\Omega\), the Poincar\'e lemma implies that a smooth closed \(1\)-form is exact, and exactness is equivalent to path-independent line integrals. Thus fixed-factor path-independence is equivalent to \(d\omega_f=0\), which is equivalent to \(\partial_{gf}F\equiv 0\) for all \(g\neq f\). Because \(\Omega\) is a product domain, this condition implies that \(\partial_fF\) is a function of \(x_f\) alone, so \(\omega_f\) is the differential of a one-variable primitive.

Applying the fixed-factor statement to every \(f\) gives the global equivalence with vanishing mixed partials. If all mixed partials vanish, then each \(\partial_fF\) depends only on \(x_f\); integrating the coordinate derivatives yields \(F(x)=c+\sum_f\Phi_f(x_f)\). The reverse implication is immediate from differentiating the additive representation.
\end{proof}

Theorem~\ref{thm:interaction-curvature} sharpens the nonlinear message. Path semantics are not required merely because the model is nonlinear; they are required precisely when mixed-partial interaction curvature renders the factor forms non-exact. In the additive/no-interaction regime, factorwise pathwise uniqueness is recovered.

\subsection{Direct attribution of the pairwise margin function}

The appropriate explanatory target is the pairwise margin itself, rather than two absolute scores explained separately and subsequently subtracted. Define the pairwise margin function
\begin{equation}
    G(u,v) = F(u)-F(v), \qquad (u,v)\in\R^{2d}.
\end{equation}
Direct attribution of \(G\) respects the pairwise semantics of ranking. By contrast, separately attributing \(F(u)\) and \(F(v)\) and then subtracting the results may introduce baseline choices and normalization artifacts that are irrelevant to the pairwise question.

Alternative attribution semantics, including SHAP-style additive explanations \citep{lundberg2017unified,shapley1953value}, can also be used at the pairwise level. The framework only requires a complete pairwise decomposition. Theorem~\ref{thm:interaction-curvature} explains why the semantic choice cannot generally be eliminated: outside the additive/no-interaction regime, factorwise path uniqueness fails even though the scalar pairwise margin remains exact.

\subsection{Pairwise boundary geometry in joint feature space}

The score-space boundary \(s_i=s_j\) remains linear in score coordinates, but its preimage in feature space may be nonlinear. Using the pairwise margin function \(G\), define the pairwise boundary
\begin{equation}
    \mathcal{B} = \{(u,v)\in \R^{2d} : G(u,v)=0\}.
\end{equation}
When \(\nabla G(u,v)\neq 0\), this is a smooth hypersurface whose normal vector is
\begin{equation}
    \nabla G(u,v) =
    \begin{bmatrix}
        \nabla F(u) \\
        -\nabla F(v)
    \end{bmatrix}.
\end{equation}
Thus, the geometric picture persists in the nonlinear setting: pairwise margin remains a normal-coordinate object, but the relevant boundary is now generally curved in joint feature space.

\subsection{Influence Exchange under chosen attribution semantics}

Suppose a chosen pairwise attribution rule produces factor-level quantities \(\widetilde{\Delta}_{ij}^{(f)}\) such that
\[
    \sum_{f=1}^d \widetilde{\Delta}_{ij}^{(f)} = \Delta_{ij}.
\]
Then the rest of the framework proceeds unchanged:
\begin{equation}
    \widetilde{\rho}_{ij}^{(f)} =
    \frac{|\widetilde{\Delta}_{ij}^{(f)}|}{\sum_k |\widetilde{\Delta}_{ij}^{(k)}|},
    \qquad
    \widetilde{I}_f = \E_{(i,j)\sim\Pcal}\left[\widetilde{\rho}_{ij}^{(f)}\right].
\end{equation}
Outside the additive/no-interaction regime of Theorem~\ref{thm:interaction-curvature}, the additional ingredient is semantic disclosure: one must state the chosen attribution rule and the factor granularity to which it is applied.

\subsection{Parameter updates, boundary crossing, and local repricing}

For a parameterized nonlinear model \(F_\theta\), a small parameter update \(\delta\theta\) changes the pairwise margin as
\begin{equation}
    \Delta_{ij}(\theta+\delta\theta)
    =
    \Delta_{ij}(\theta)
    +
    \left\langle
    \nabla_\theta F_\theta(x_i) - \nabla_\theta F_\theta(x_j),
    \delta\theta
    \right\rangle
    +
    O(\|\delta\theta\|^2).
\end{equation}
This differential view suggests two mechanisms for nonlinear Influence Exchange. First, a parameter update may push a pairwise margin across zero, producing an actual boundary crossing. Second, even without an immediate flip, it may alter local gradients and thereby reprice factor-level path contributions. Such updates can also change the mixed-partial field \(\partial_{gf}F_\theta\), changing the degree of factorwise path dependence characterized by Theorem~\ref{thm:interaction-curvature}. We treat this as a local interpretation of parameter perturbations, not as a complete dynamical theory of training.

\section{Exactness, Cycle Consistency, and Hodge-Like Views}

The framework thus far has treated pairwise margins as local objects associated with pairs. This section addresses a more structural question: under what conditions can an antisymmetric pairwise field be generated by a global score function, and what changes when factor-level decompositions cease to be individually exact?

Theorem~\ref{thm:interaction-curvature} identified mixed-partial interaction curvature as the continuous source of factorwise path dependence. We now examine the discrete edge-field shadow of the same phenomenon on finite item sets. After one passes from smooth feature-space paths to sampled pairwise comparisons, unresolved factorwise interaction tension appears as cycle inconsistency and, locally, as triangle curl.

\subsection{Edge space versus score-generated subspace}

Let \(A=(A_{ij})\) be an antisymmetric pairwise field on \(n\) items, so \(A_{ij}=-A_{ji}\). The space of all such fields has dimension \(\binom{n}{2}\). By contrast, score-generated fields have the form
\begin{equation}
    A_{ij} = s_i - s_j
\end{equation}
for some \(s\in\R^n\), and this space has dimension only \(n-1\) because \(s\) is defined only up to an additive constant.

\begin{proposition}[Dimension gap]
The codimension of the score-generated subspace inside the antisymmetric edge space is
\begin{equation}
    \binom{n}{2}-(n-1)=\frac{(n-1)(n-2)}{2}.
\end{equation}
\end{proposition}

\begin{proof}
Subtracting dimensions gives the stated formula. The \(n-1\) dimension of score-generated fields follows from quotienting out the gauge direction \(s\mapsto s+c\one\).
\end{proof}

The significance is conceptual: score-generated ranking fields do not occupy the generic pairwise world. They reside in a highly constrained subspace thereof.

\subsection{Score representability and cycle consistency}

\begin{theorem}[Score representability]
Let \(A\) be an antisymmetric field on the complete graph over \(n\) items. Then the following are equivalent:
\begin{enumerate}[leftmargin=2em]
    \item There exists \(s\in\R^n\) such that \(A_{ij}=s_i-s_j\) for all \(i,j\).
    \item \(A\) is cycle consistent, that is,
    \begin{equation}
        A_{i_0i_1} + A_{i_1i_2} + \cdots + A_{i_{m-1}i_0}=0
    \end{equation}
    for every cycle \((i_0,\dots,i_{m-1},i_0)\).
\end{enumerate}
\end{theorem}

This is a standard exactness statement for pairwise comparison fields and is closely related to HodgeRank-style formulations \citep{jiang2011statistical,lim2020hodge}. A proof is included in Appendix~A for completeness.

\begin{remark}[A discrete differential viewpoint]
If \(s\) is viewed as a \(0\)-cochain on vertices, then \(A=ds\) is an exact \(1\)-cochain on edges. We use this language as a lens rather than as a full topological development.
\end{remark}

\subsection{Factor-level exactness in the linear case}

\begin{proposition}[Linear factor fields are exact]
In the linear model, the factor-level field
\begin{equation}
    A_{ij}^{(f)} = \Delta_{ij}^{(f)} = w_f(x_{if}-x_{jf})
\end{equation}
is exact for each factor \(f\). In particular, \(A^{(f)} = d\phi_f\) with \(\phi_f(i)=w_f x_{if}\).
\end{proposition}

\begin{proof}
By definition,
\[
    \phi_f(i)-\phi_f(j) = w_f x_{if} - w_f x_{jf} = w_f(x_{if}-x_{jf}) = A_{ij}^{(f)}.
\]
\end{proof}

This is a principal reason the linear case is especially transparent: not only the total field, but each factor field individually, is integrable.

\subsection{Factor-level exactness can break in nonlinear accounting}

In nonlinear models, the total pairwise field remains exact because it is still generated by scalar scores \(F(x_i)\). However, Theorem~\ref{thm:interaction-curvature} shows that factorwise path attribution is canonically path-independent only in the additive/no-interaction regime. After a path semantics is chosen and evaluated on a finite item set, the induced factor-level edge field need not preserve exactness for each individual factor.

\begin{definition}[Triangle curl]
For a factor-level field \(A^{(f)}\), define the triangle curl
\begin{equation}
    \kappa_f(i,j,k) = A_{ij}^{(f)} + A_{jk}^{(f)} + A_{ki}^{(f)}.
\end{equation}
\end{definition}

\begin{remark}[Exactness test]
If \(A^{(f)}\) is exact, then \(\kappa_f(i,j,k)=0\) for every triple \((i,j,k)\). Conversely, nonzero triangle curl witnesses a local failure of factor-level exactness.
\end{remark}

Section~7 identified mixed-partial interaction curvature as the continuous source of factorwise path dependence. The triangle curl introduced here is a discrete edge-field diagnostic of the same obstruction after one has chosen an attribution semantics and sampled pairwise comparisons on a finite item set.

\begin{remark}[Cancellation in the total field]
If the total field \(A=\sum_f A^{(f)}\) is exact, then for every triple,
\begin{equation}
    \sum_f \kappa_f(i,j,k)=0.
\end{equation}
Consequently, nonlinear factor-level cyclic tensions may exist, but they must cancel in the total field.
\end{remark}

\begin{figure}[t]
    \centering
    \includegraphics[width=0.62\linewidth]{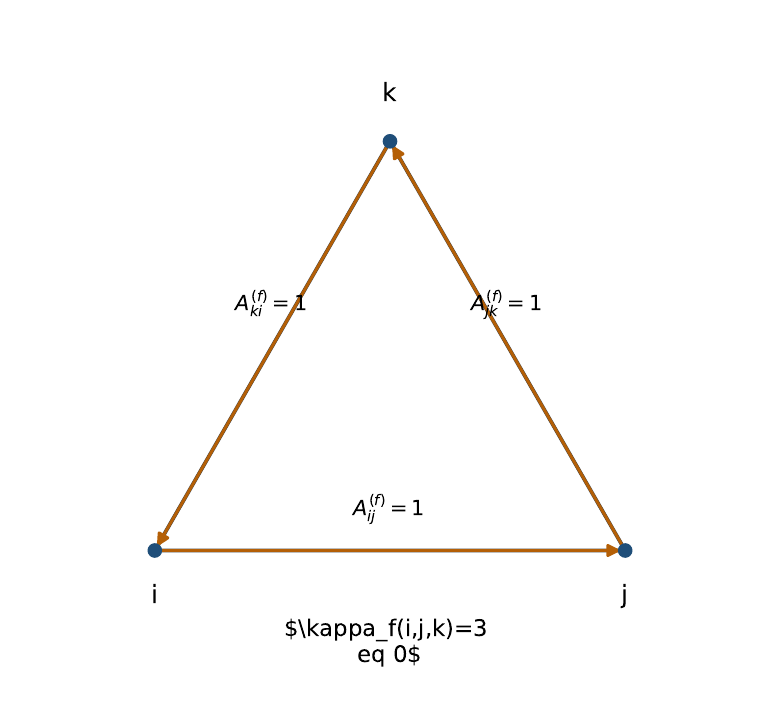}
    \caption{A factor-level triangle field with nonzero cyclic tension. The total ranking field may still be exact after cancellation across factors.}
    \label{fig:triangle-curl}
\end{figure}

\subsection{Hodge-like decomposition as a diagnostic lens}

On a finite sampled pairwise graph, one may further project a factor-level field onto an exact component and a non-exact residual using graph-based Hodge machinery \citep{jiang2011statistical,lim2020hodge}. In schematic form,
\begin{equation}
    A^{(f)} = d\phi_f + r_f,
\end{equation}
where \(d\phi_f\) is the exact component and \(r_f\) denotes the remaining non-exact structure. Depending on the sampled graph, clique complex, inner product, and projection convention, this residual may include curl and harmonic components. We employ this solely as a diagnostic lens. The present paper does not develop the full operator-theoretic setting, inner products, simplicial-complex choice, or estimation procedures required for a complete Hodge theory of Influence Exchange.

Nevertheless, the interpretation is valuable. The exact component measures factor-level influence that can be reconciled with a globally score-compatible accounting, while the non-exact residual measures finite-graph factor-level inconsistency that cannot be integrated into a factor-level potential without leftover structure. This viewpoint provides a natural transition to the geometric closure of the following section.

\section{Root Space, Coxeter Arrangements, and Weyl Chambers}

The preceding sections connected ranking to pairwise margins, permutation changes, hyperplane crossings, and exact edge fields. This section completes the geometric account by quotienting out the irrelevant global-shift direction and situating rankings within the standard type-\(A\) root-space framework.

\subsection{Gauge fixing}

Pairwise margins are invariant under the transformation \(s \mapsto s + c\one\). It is therefore natural to fix the gauge by requiring
\begin{equation}
    \sum_{i=1}^n s_i = 0.
\end{equation}
The effective score space becomes the \((n-1)\)-dimensional hyperplane
\begin{equation}
    \mathfrak{h} = \{s\in\R^n : \langle s,\one\rangle = 0\},
\end{equation}
which is the standard root space for type \(A_{n-1}\) \citep{humphreys1990reflection,bjornerbrenti2005}.

\subsection{Roots and chambers}

For each pair \(i \neq j\), define the root vector
\begin{equation}
    \alpha_{ij} = \e_i-\e_j.
\end{equation}
Then
\begin{equation}
    \Delta_{ij} = s_i-s_j = \langle s,\alpha_{ij}\rangle.
\end{equation}
The boundary hyperplanes are precisely
\begin{equation}
    H_{ij} = \{s\in\mathfrak{h} : \langle s,\alpha_{ij}\rangle = 0\}.
\end{equation}
This is the type-\(A\) Coxeter arrangement. Its chambers are Weyl chambers, and each chamber corresponds to one strict ordering of the coordinates. Equivalently, each chamber corresponds to one permutation of the items.

\begin{figure}[t]
    \centering
    \includegraphics[width=0.72\linewidth]{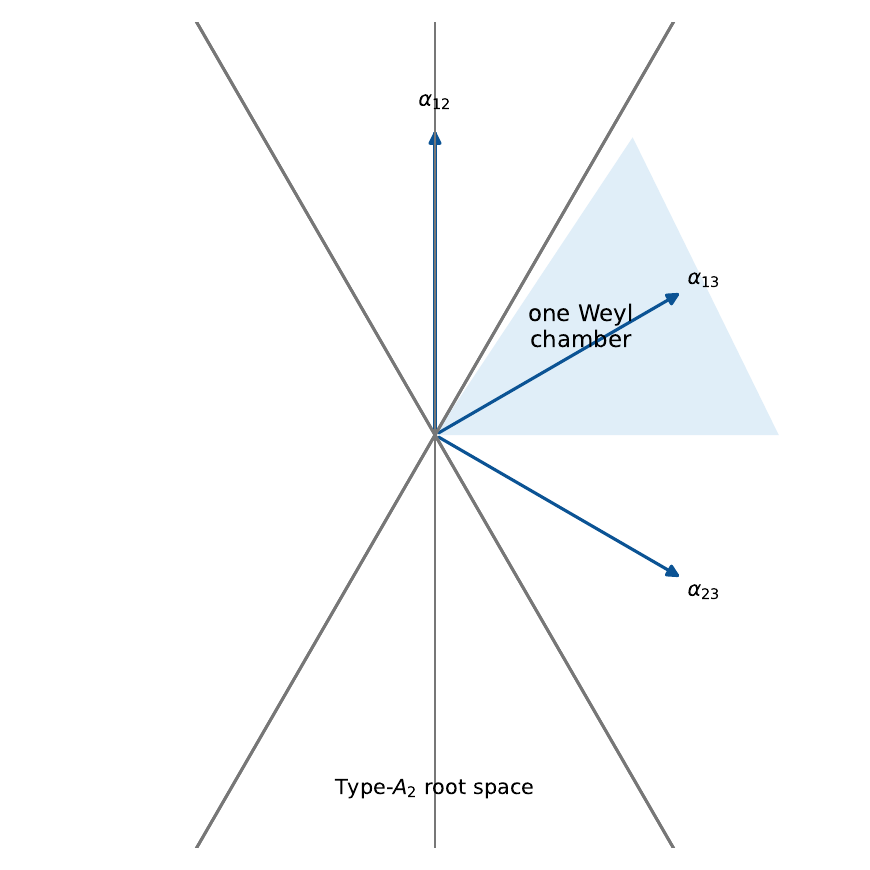}
    \caption{A conceptual type-\(A_2\) root-space picture. Root directions determine boundary mirrors, and each Weyl chamber corresponds to one total order.}
    \label{fig:root-space}
\end{figure}

\subsection{Interpretation}

This viewpoint does not alter the formal objects of the paper; it reorganizes them. A ranking corresponds to the chamber containing the gauge-fixed score vector. A local flip corresponds to crossing one mirror of the arrangement. A pairwise margin is the coordinate of the score vector with respect to a root vector. Influence Exchange then tracks how global factor-level control evolves as the score vector moves through a geometry whose chamber structure is indexed by permutations.

The permutahedron provides a related combinatorial picture: the braid arrangement above is the normal fan of the permutahedron \citep{postnikov2009permutohedra}. We mention this connection solely for geometric completeness, not to introduce a new central object. The manuscript remains centered on pairwise margins and Influence Exchange.

\section{Related Work}

This section positions the manuscript relative to existing literature. The goal is not to survey ranking, explainability, or geometric combinatorics comprehensively, but to situate the specific pairwise influence-accounting perspective developed in the preceding sections.

\paragraph{Learning to rank and pairwise ranking.}
Pairwise ranking is classical in information retrieval and preference learning. Large-margin and clickthrough-based learning-to-rank methods use pairwise preferences to learn ranking functions \citep{joachims2002optimizing}, while RankNet and related models center pairwise comparisons in neural learning objectives \citep{burges2005ranknet,burges2010ranknet,wu2010lambdamart}. Bradley--Terry and Thurstone models provide earlier statistical formulations of pairwise comparison \citep{bradley1952rank,thurstone1927law}. The present paper does not introduce a new loss function or estimator in this line. Its contribution is to formalize a pairwise \emph{influence-accounting} view centered on margins, local shares, global structure, and exchange.

\paragraph{Feature attribution and path methods.}
The nonlinear section is informed by attribution methods such as Integrated Gradients, Aumann--Shapley-style value decompositions, and SHAP \citep{sundararajan2017axiomatic,aumann1974values,lundberg2017unified,shapley1953value}. These works address scalar model outputs or additive value allocations more broadly. Our contribution is to place attribution directly on the pairwise margin function and to use the resulting factor-level pairwise contributions as the basis for influence shares and Influence Exchange.
Work on feature-interaction attribution also uses mixed-derivative information to explain interactions among inputs \citep{janizek2021explaining}. Our use of mixed partials is different: we isolate a path-independence boundary for the factorwise margin ledger. On the product-domain smooth setting of Theorem~\ref{thm:interaction-curvature}, factorwise path uniqueness is equivalent to vanishing mixed partials, and full factorwise uniqueness is equivalent to additive separability.

\paragraph{Permutation distance and pairwise consistency.}
Kendall tau and adjacent-swap characterizations are standard tools for measuring disagreement between rankings \citep{kendall1938new,bjornerbrenti2005}. They count pairwise reversals but do not attribute the forces underlying those reversals. On the representability side, HodgeRank and related work connect pairwise comparison data to exactness, cyclic inconsistency, and graph-based decompositions \citep{jiang2011statistical,lim2020hodge}. These mathematical structures provide a natural language for the exactness and Hodge-like components of our framework.
An additional contribution of this manuscript is the bridge from continuous interaction obstruction to discrete non-exact residual: mixed-partial curvature explains path dependence before sampling; triangle curl diagnoses a chosen factor-level edge-field after sampling.

\paragraph{Production influence optimization.}
A companion Sortify technical report uses Influence Share/Exchange as the accounting layer of an autonomous LLM-driven ranking-optimization agent and describes online deployment in a large-scale production recommendation system \citep{cheng2026letagentsteer}. That report is the appropriate source for the system architecture, online A/B protocol, market-specific deployment narrative, and business-metric outcomes. The present paper is complementary: it abstracts from the deployed system and studies the pairwise-margin, local-share, exchange-geometry, nonlinear-attribution, exactness, and chamber-geometry foundations of the same influence-accounting language.

\paragraph{Hyperplane geometry, root systems, and Weyl chambers.}
The geometry of score-induced rankings belongs to the classical type-\(A\) arrangement, reflection-group, and root-system literature \citep{humphreys1990reflection,stanley2007hyperplane,postnikov2009permutohedra}. We do not claim novelty for these structures per se. Our observation is that they become especially natural once ranking is centered on pairwise margins and Influence Exchange, since the same pairwise boundary hyperplanes that define chambers also define the local margin coordinates on which the accounting framework is constructed.

\section{Scope, Limitations, and Claim Boundaries}

This paper intentionally separates proved statements, standard recalled results, and higher-level structural interpretations. Table~\ref{tab:claim-ladder} summarizes this stratification.

\begin{table}[t]
    \centering
    \caption{Claim ladder used in the manuscript.}
    \label{tab:claim-ladder}
    \begin{tabularx}{0.96\linewidth}{>{\raggedright\arraybackslash}p{0.16\linewidth} >{\raggedright\arraybackslash}p{0.22\linewidth} >{\raggedright\arraybackslash}X}
\toprule
Layer & Status in the paper & Typical examples \\
\midrule
Layer A & Proved directly in the main text & Linear margin decomposition, refinement-consistency characterization of local share, exchange geometry / Laplacian response / zero-exchange rigidity, margin as normal coordinate, linear recovery of Pairwise IG, interaction-curvature characterization of factorwise path dependence \\
Layer B & Recalled as standard, sketched, or proved in the appendix & Kendall distance facts, midpoint expansion, score representability, gauge fixing, chamber correspondence \\
Layer C & Interpreted structurally & Pairwise control allocation, local pricing, factor-level cyclic tension, Weyl chambers as geometric home \\
Layer D & Explicitly limited & No causal interpretation, no unique nonlinear attribution outside the additive/no-interaction regime, no theorem-heavy novelty claim for Hodge-like or Weyl views \\
\bottomrule
\end{tabularx}

\end{table}

\paragraph{Empirical evidence and scope.}
This manuscript is not itself an experimental report: it contains no new benchmark suite, ablation study, or independent deployment-study section. This limitation should not be read as saying that Influence Share/Exchange lacks production evidence. The companion Sortify technical report \citep{cheng2026letagentsteer} describes a fully autonomous LLM-driven ranking-optimization agent deployed in a large-scale production recommendation system, with online A/B validation, positive business-metric outcomes, and subsequent production rollout. The present paper supplies a formal pairwise-margin and influence-accounting foundation rather than a duplicate empirical evaluation of that production system.

\paragraph{Ties, degenerate pairs, and pair distributions.}
The main body assumes no ties for clarity. Influence shares are defined only on informative pairs with \(Z_{ij}>0\). Both restrictions are pragmatic simplifications rather than fundamental limitations, and Appendix~C addresses the corresponding edge cases. Similarly, the pair distribution \(\Pcal\) is specified by the analyst; the framework does not prescribe a universal weighting scheme.

\paragraph{Canonical local shares do not remove all modeling choices.}
For linear local effort vectors, Theorem~\ref{thm:refinement-local-share} proves the uniqueness of the refinement-consistent budgeting rule used for \(\rho_{ij}^{(f)}\). This does not determine the analyst's item-pair distribution \(\Pcal\), the semantic definition of a factor, or the population on which global influence is averaged. The log-weight rigidity result also excludes zero-weight or never-varying factors into the inactive complement of \(\mathcal F_*\); these coordinates can be reported as trivial zero-share coordinates, but they are not identified by the log-absolute reparameterization.

\paragraph{Nonlinear attribution is semantic, not unique truth.}
In nonlinear models, factor-level accounting generally depends on a chosen attribution semantics and on a chosen factor granularity. Theorem~\ref{thm:interaction-curvature} identifies the precise exceptional regime where factorwise path-independence is recovered: the relevant mixed partials vanish, and full factorwise uniqueness corresponds to additive separability. Outside that regime, Pairwise Integrated Gradients is one complete path-based choice that recovers the linear case, but it is not the only admissible semantics. Nothing in this paper should be interpreted as a claim that nonlinear path attribution reveals a causal or universally correct decomposition.

\paragraph{Exactness and Hodge-like views are scoped.}
The exactness discussion in Section~8 is rigorous for total score-generated fields and for linear factor fields. The Hodge-like decomposition is discussed only as a structural diagnostic lens applied to finite sampled pairwise graphs. We do not develop a complete statistical estimation theory, orthogonality framework, or robustness theory for these decompositions here.

\paragraph{Weyl chambers are an interpretive closure.}
The root-space and Weyl-chamber section is included because it provides a natural geometric setting for score-induced total orders after gauge fixing. It is not advanced as an independent novelty claim. The manuscript employs this geometry to complete the conceptual chain from pairwise margins to chamber-valued ranking outcomes.

\section{Conclusion}

This paper has developed a pairwise-first theory of ranking. Absolute scores serve as implementation coordinates; pairwise margins constitute the atomic analytical objects. In the linear case, exact factor-level local contributions lead to a canonical local influence share: refinement-consistent local budgeting forces the \(L_1\) share. Upon aggregation, the global influence structure is not merely a simplex-valued summary; in log-absolute-weight coordinates it is a convex-potential gradient with a competition-graph Laplacian response. Influence Exchange is therefore a structured reallocation of pairwise control, with zero-exchange rigidity modulo the common-scaling gauges of the competition graph.

The same perspective extends to nonlinear scoring once the invariant object is distinguished from the semantic choice. The invariant object is the pairwise margin. The semantic choice concerns how a nonlinear margin is decomposed into factor-level path contributions. The interaction-curvature theorem sharpens this boundary: path semantics are needed exactly to manage mixed-partial interaction curvature, and factorwise pathwise uniqueness is recovered precisely in the additive/no-interaction regime.

At a higher structural level, permutation distance, hyperplane crossings, exact antisymmetric edge fields, factor-level curl, and root-space and Weyl-chamber geometry are organized as successive closures of the same analytical progression. The continuous interaction obstruction of Section~7 becomes the discrete triangle-curl and cyclic-residual language of Section~8; the pairwise boundary picture then closes in the type-\(A\) chamber geometry of Section~9. Future work may test empirical diagnostics of exchange rigidity, investigate robustness under alternative nonlinear semantics and factor granularity, and develop graph-based estimators for the exact/cyclic decomposition after the continuous-to-discrete bridge established here.

\appendix

\section{Additional Proofs}

\subsection{Midpoint local linearization}

We prove Proposition~7.1. Let
\[
    \phi(t) = F(\barx + t\deltax).
\]
Then
\[
    \Delta_{ij} = \phi\!\left(\frac{1}{2}\right)-\phi\!\left(-\frac{1}{2}\right).
\]
Expanding \(\phi\) around \(0\) gives
\[
    \phi(t)
    =
    \phi(0)
    +
    t\phi'(0)
    +
    \frac{t^2}{2}\phi''(0)
    +
    \frac{t^3}{6}\phi'''(\xi_t)
\]
for some \(\xi_t\) between \(0\) and \(t\). Subtracting the expansions at \(t=1/2\) and \(t=-1/2\) cancels the even terms, so
\[
    \Delta_{ij}
    =
    \phi'(0) + O(\|\deltax\|^3).
\]
By the chain rule, \(\phi'(0)=\nabla F(\barx)^\top \deltax\), proving the claim.

\subsection{Proof of score representability}

We prove Theorem~8.2.

\paragraph{(1) implies (2).}
If \(A_{ij}=s_i-s_j\), then for any cycle \((i_0,\dots,i_{m-1},i_0)\),
\[
    A_{i_0i_1}+\cdots + A_{i_{m-1}i_0}
    =
    (s_{i_0}-s_{i_1}) + \cdots + (s_{i_{m-1}}-s_{i_0}) = 0
\]
by telescoping cancellation.

\paragraph{(2) implies (1).}
Fix a reference node \(r\). Define
\[
    s_i = A_{ir}.
\]
For any pair \(i,j\), cycle consistency on the triangle \((i,j,r,i)\) yields
\[
    A_{ij} + A_{jr} + A_{ri}=0.
\]
Using antisymmetry, \(A_{ri}=-A_{ir}\), so
\[
    A_{ij} = A_{ir} - A_{jr} = s_i - s_j.
\]
Hence \(A\) is score representable.

\subsection{Triangle curl cancellation}

Suppose \(A=\sum_f A^{(f)}\) and the total field \(A\) is exact. Then for every triple \((i,j,k)\),
\[
    0
    =
    A_{ij}+A_{jk}+A_{ki}
    =
    \sum_f \left(A_{ij}^{(f)}+A_{jk}^{(f)}+A_{ki}^{(f)}\right)
    =
    \sum_f \kappa_f(i,j,k).
\]
This establishes the cancellation statement used in Section~8.

\section{Toy Examples}

\subsection{A linear two-factor example}

Consider two items \(i\) and \(j\) and two factors. Let
\[
    \Delta_{ij}^{(1)} = 3,
    \qquad
    \Delta_{ij}^{(2)} = -2.
\]
Then
\[
    \Delta_{ij} = 1,
    \qquad
    Z_{ij} = |3| + |-2| = 5,
    \qquad
    \rho_{ij}^{(1)} = \frac{3}{5},
    \qquad
    \rho_{ij}^{(2)} = \frac{2}{5}.
\]
This example illustrates the distinction between direction and local participation. Factor \(1\) supports item \(i\), factor \(2\) supports item \(j\), and both contribute nontrivially to the local pairwise contest.

\subsection{Path semantics in the bilinear model}

Take \(F(x_1,x_2)=x_1x_2\) and two points \(x_j=(a,b)\), \(x_i=(c,d)\).

\paragraph{Axis-aligned path.}
Move first in the \(x_1\) direction and then in the \(x_2\) direction. The resulting factor contributions are
\[
    \Delta_{ij}^{\gamma_1,(1)} = b(c-a),
    \qquad
    \Delta_{ij}^{\gamma_1,(2)} = c(d-b).
\]

\paragraph{Reverse axis-aligned path.}
Move first in the \(x_2\) direction and then in the \(x_1\) direction. Then
\[
    \Delta_{ij}^{\gamma_2,(1)} = d(c-a),
    \qquad
    \Delta_{ij}^{\gamma_2,(2)} = a(d-b).
\]

\paragraph{Straight-line path.}
Pairwise Integrated Gradients gives
\[
    \Delta_{ij}^{\mathrm{PIG},(1)} = \frac{b+d}{2}(c-a),
    \qquad
    \Delta_{ij}^{\mathrm{PIG},(2)} = \frac{a+c}{2}(d-b).
\]

All three allocations sum to the same total margin \(cd-ab\), but the factor ledger changes with the path semantics.

\subsection{Direct attribution of the pairwise margin}

Let
\[
    F(x_1,x_2,x_3) = x_1x_2 + x_3.
\]
Then the pairwise target is
\[
    G(u,v) = F(u)-F(v) = (u_1u_2-v_1v_2) + (u_3-v_3).
\]
Directly attributing \(G\) keeps the explanatory object aligned with the ranking decision. By contrast, separately attributing \(F(u)\) and \(F(v)\) may require independent baseline conventions whose difference is not itself a primitive ranking quantity.

\section{Notation and Edge Cases}

\subsection{Notation summary}

\begin{table}[h]
    \centering
    \caption{Core notation used throughout the manuscript.}
    \label{tab:notation}
    \begin{tabularx}{0.96\linewidth}{>{\raggedright\arraybackslash}p{0.26\linewidth} >{\raggedright\arraybackslash}X}
\toprule
Object & Definition \\
\midrule
Item indices & \(i,j,k\) \\
Pair index in exchange geometry & \(p=(i,j)\) \\
Feature vector & \(x_i \in \R^d\) \\
Score & \(s_i\) \\
Pairwise margin & \(\Delta_{ij}=s_i-s_j\) \\
Linear factor contribution & \(\Delta_{ij}^{(f)} = w_f(x_{if}-x_{jf})\) \\
Informative-pair effort & \(Z_{ij} = \sum_f |\Delta_{ij}^{(f)}|\) \\
Local influence share & \(\rho_{ij}^{(f)} = |\Delta_{ij}^{(f)}|/Z_{ij}\) \\
Pair distribution & \((i,j)\sim\Pcal\) \\
Global influence structure & \(I_f=\E_{(i,j)\sim\Pcal}[\rho_{ij}^{(f)}]\) \\
Influence Exchange & \(\Exch(\theta,\theta')=\Ivec(\theta')-\Ivec(\theta)\) \\
Absolute pairwise feature gap & \(b_p^{(f)}=|x_{if}-x_{jf}|\), used in the log-weight exchange theorem \\
Active factor set & \(\mathcal F_*\), the nonzero, pair-varying factors in Theorem~\ref{thm:exchange-geometry} \\
Log-absolute linear weight & \(u_f=\log |w_f|\), local to Theorem~\ref{thm:exchange-geometry} \\
Exchange potential and response & \(\Phi(u)\) and \(J(u)=\nabla^2\Phi(u)\) in the exchange-geometry theorem \\
Nonlinear pairwise margin & \(\Delta_{ij}=F(x_i)-F(x_j)\) \\
Path-based contribution & \(\Delta_{ij}^{\gamma,(f)}\) \\
Factorwise path form & \(\omega_f=\partial_fF(x)\,dx_f\), used in Theorem~\ref{thm:interaction-curvature} \\
Pairwise margin function & \(G(u,v)=F(u)-F(v)\) \\
Antisymmetric edge field & \(A_{ij}=-A_{ji}\) \\
Triangle curl & \(\kappa_f(i,j,k)=A_{ij}^{(f)}+A_{jk}^{(f)}+A_{ki}^{(f)}\) \\
Root vector & \(\alpha_{ij}=\e_i-\e_j\) \\
\bottomrule
\end{tabularx}

\end{table}

\subsection{Ties and degenerate local contests}

The main body assumes \(s_i\neq s_j\) for all distinct \(i,j\). In applications, ties may occur because of exact equality, quantization, or deterministic tie-breaking rules. The pairwise-first framework can still be used, but one must specify whether ties are excluded, broken by an external rule, or treated as a separate outcome class.

Likewise, local influence shares are defined only when
\[
    Z_{ij} = \sum_f |\Delta_{ij}^{(f)}| > 0.
\]
If \(Z_{ij}=0\), then every factor-level local contribution is zero and the pair carries no information for share allocation. We therefore exclude such pairs from the local share definition.

\subsection{Factor granularity}

The framework assumes a chosen set of factors. In a linear model, a factor can be one feature, one feature group, or one business term in a larger decomposition. Theorem~\ref{thm:refinement-local-share} shows that if such a factor is only refined into a bookkeeping split of nonnegative local effort, the local share must split proportionally and the block share is preserved. In nonlinear models the granularity choice is even more consequential because interaction terms may be grouped or split in different ways. Influence Exchange is therefore conditional on an analyst-defined factorization.

\subsection{Zero coordinates in log-weight exchange geometry}

The log-weight coordinates \(u_f=\log |w_f|\) used in Theorem~\ref{thm:exchange-geometry} are defined only for factors with nonzero linear weights. The theorem therefore restricts to \(\mathcal F_*\): factors with \(|w_f|>0\) that vary on a set of pairs with positive \(\Pcal\)-probability. Zero-weight factors or factors that never vary under \(\Pcal\) remain valid coordinates in the original accounting pipeline, but they enter the log-parameterized rigidity theorem only as trivial zero-share coordinates outside \(\mathcal F_*\).

\subsection{Path choice in nonlinear attribution}

Path-based pairwise attribution requires a chosen path or an equivalent attribution semantics outside the additive/no-interaction regime characterized by Theorem~\ref{thm:interaction-curvature}. Straight-line paths are attractive because they are simple, complete, and recover the linear case. They are not mandatory. Different paths encode different views of how interaction effects should be distributed across factors.

\subsection{Gauge fixing}

Because pairwise margins are invariant under global score shifts, one may work either in \(\R^n\) modulo the span of \(\one\) or in the gauge-fixed hyperplane \(\sum_i s_i=0\). The latter viewpoint is used in Section~9 because it aligns naturally with root-space geometry.

\section{Additional Geometry and Graph Background}

\subsection{The braid arrangement}

The hyperplanes \(H_{ij}=\{s_i=s_j\}\) form the braid arrangement of type \(A_{n-1}\). Restricting to the gauge-fixed space \(\sum_i s_i=0\) removes the global-shift redundancy and leaves an arrangement whose chambers correspond bijectively to total orders of the \(n\) coordinates.

\subsection{Weyl chambers}

In the type-\(A\) root system, the chambers cut out by the roots \(\alpha_{ij}=\e_i-\e_j\) are Weyl chambers. Each chamber is specified by a strict chain of inequalities among the coordinates, such as
\[
    s_{\pi(1)} > s_{\pi(2)} > \cdots > s_{\pi(n)}.
\]
Hence the chamber labels are permutations.

\subsection{Finite-graph Hodge viewpoint}

For a finite sampled pairwise graph \(G=(V,E)\), an antisymmetric edge field assigns a signed value to each oriented edge. Classical graph Hodge theory decomposes such a field into gradient-like, curl-like, and harmonic components under a chosen inner product. In ranking applications, the exact component captures score-compatible structure while the cyclic component captures local inconsistency \citep{jiang2011statistical,lim2020hodge}. The manuscript borrows this language only at the level needed to interpret factor-level cyclic tension.

\bibliographystyle{plainnat}
\bibliography{refs}

\end{document}